\newif\ifarxiv
\arxivtrue

\documentclass[a4paper,USenglish,cleveref, autoref, thm-restate]{lipics-v2021}
\ifarxiv{}
  \nolinenumbers
  \hideLIPIcs
\fi{}

\usepackage{mathtools}
\theoremstyle{definition}

\usepackage{xspace}
\usepackage[textwidth=3.5cm]{todonotes}
\usepackage[sort,numbers]{natbib}
\usepackage{sidecap}
\sidecaptionvpos{figure}{t}
\usepackage{tikz}
\tikzset{
	vertex/.style={circle,draw,fill=black, minimum size = 5pt, inner sep=0pt},
	edge/.style={color=black,ultra thick},
	diredge/.style={->,color=black, ultra thick},
}
\usepackage[algo2e,ruled,vlined,linesnumbered,hanginginout]{algorithm2e}
\SetEndCharOfAlgoLine{}
\let\oldnl\nl%
\newcommand{\nonl}{\renewcommand{\nl}{\let\nl\oldnl}}%
\renewcommand{\thealgocf}{\arabic{theorem}}

\usepackage{doi}
\title{Temporal Reachability Minimization:\\ Delaying vs.\ Deleting}
\titlerunning{Temporal Reachability Minimization: Delaying vs.\ Deleting} %

\newcommand{\tuaddress}{Technische Universität Berlin, Faculty IV, Algorithmics and Computational Complexity, Germany}

\author{Hendrik Molter}
{Department of Industrial Engineering and Management, Ben-Gurion University of the Negev, 
Israel}
{molterh@post.bgu.ac.il}
{https://orcid.org/0000-0002-4590-798X}{Supported by the German Research
Foundation (DFG), project MATE (NI 369/17), and by the Israeli Science Foundation (ISF), grant No.~1070/20.
 The main work was done while affiliated with TU Berlin.}

\author{Malte Renken}
{\tuaddress}
{m.renken@tu-berlin.de}
{http://orcid.org/0000-0002-1450-1901}{Supported by the German Research
Foundation (DFG), project MATE (NI 369/17).}

\author{Philipp Zschoche}
{\tuaddress}
{zschoche@tu-berlin.de}
{https://orcid.org/0000-0001-9846-0600}{}

\authorrunning{H.~Molter, M.~Renken, and P.~Zschoche} %

\Copyright{Hendrik Molter, Malte Renken, and Philipp Zschoche} %



\category{} %

\relatedversion{} %

\funding{}%

\acknowledgements{This work was initiated at the research retreat of the
Algorithmics and Computational Complexity group of TU Berlin in September 2020
in Zinnowitz.}

\EventEditors{John Q. Open and Joan R. Access}
\EventNoEds{2}
\EventLongTitle{42nd Conference on Very Important Topics (CVIT 2016)}
\EventShortTitle{CVIT 2016}
\EventAcronym{CVIT}
\EventYear{2016}
\EventDate{December 24--27, 2016}
\EventLocation{Little Whinging, United Kingdom}
\EventLogo{}
\SeriesVolume{42}
\ArticleNo{23}

\theoremstyle{definition}

\newcommand{\problemdef}[3]{
	\begin{center}
	\begin{minipage}{0.96\textwidth}
		\noindent
		#1
		\vspace{4pt}\\
		\setlength{\tabcolsep}{3pt}
		\begin{tabularx}{\textwidth}{@{}lX@{}}
			\textbf{Input:}     & #2 \\
			\textbf{Question:}  & #3
		\end{tabularx}
	\end{minipage}
	\end{center}
}

\crefformat{equation}{(#2#1#3)}
\crefrangeformat{equation}{(#3#1#4) to~(#5#2#6)}
\crefmultiformat{equation}{(#2#1#3)}{ and~(#2#1#3)}{, (#2#1#3)}{ and~(#2#1#3)}

\DeclarePairedDelimiterX{\abs}[1]{\lvert}{\rvert}{#1}

\newcommand{\tn}{\textnormal}

\newcommand{\Wone}{\textsc{W[1]}\xspace}
\newcommand{\FPT}{\textsc{FPT}\xspace}
\newcommand{\NN}{\mathbb{N}}

\newcommand{\Tt}{\mathcal{T}}
\newcommand{\lifetime}{\ensuremath{\tau}}
\newcommand{\duration}{\gamma}
\newcommand{\reach}{R}
\newcommand{\TE}{\mathcal{E}}
\newcommand{\TG}{\mathcal{G}}
\newcommand{\TGcompact}{\ensuremath{\mathcal{G}=(V, (E_i)_{i\in[\lifetime]})}\xspace}
\newcommand{\TGtraversal}{\ensuremath{\mathcal{G}=(V, (E_i)_{i\in[\lifetime]}, \duration)}\xspace}
\newcommand{\FlowNet}{\mathcal{F}}

\newcommand{\NP}{\textrm{NP}}
\newcommand{\bigO}{{O}}

\newcommand{\ie}{i.\,e.,\ }
\newcommand{\yes}{\emph{yes}}
\newcommand{\no}{\emph{no}}
\newcommand{\TemporalDelayingLong}{\textsc{Minimizing Temporal Reachablity by Delaying}\xspace}
\newcommand{\TemporalDelaying}{\textsc{MinReachDelay}\xspace}

\newcommand{\TemporalDeletingLong}{\textsc{Minimizing Temporal Reachablity by Deleting}\xspace}
\newcommand{\TemporalDeleting}{\textsc{MinReachDelete}\xspace}

\newcommand{\TemporalSlowingLong}{\textsc{Minimizing Temporal Reachablity by Slowing}\xspace}
\newcommand{\TemporalSlowing}{\textsc{MinReachSlow}\xspace}

\newcommand{\true}{\texttt{true}}
\newcommand{\false}{\texttt{false}}

\newcommand{\mvert}{\;\middle\vert\;}

\newcommand{\delay}{\nearrow_\delta}
\newcommand{\slow}{\uparrow_\delta}

\begin{document}

\maketitle

\begin{abstract}
We study spreading processes in temporal graphs, \ie graphs whose connections change over time.
These processes naturally model real-world phenomena such as infectious diseases or information flows.
More precisely, we investigate how such a spreading process, emerging from a given set of sources, can be contained to a small part of the graph.
To this end we consider two ways of modifying the graph, which are (1) deleting connections and (2) delaying connections.
We show a close relationship between the two associated problems and give a polynomial time algorithm when the graph has tree structure.
For the general version, we consider parameterization by the number of vertices to which the spread is contained.
Surprisingly, we prove W[1]-hardness for the deletion variant but fixed-parameter tractability for the delaying variant.
\end{abstract}

\section{Introduction}
Reachability is a fundamental problem in graph theory and
algorithmics~\cite{reingold2008undirected,savitch1970relationships,hopcroft1973algorithm,dijkstra1959note}
and quite well-understood. With the emergence of temporal
graphs,\footnote{Temporal graphs are graphs whose edge set changes over discrete
time.} the concept of reachability was extended to the dimension of time using \emph{temporal
paths}~\cite{kempe2002connectivity,berman1996vulnerability}.
For a vertex~$s$ to reach another vertex~$z$ in a temporal graph,
there must not only be a path between them but the edges of this path have to appear in chronological order.
This requirement makes temporal reachability non-symmetric and non-transitive,
which stands in contrast to reachability in normal (static) graphs.
Reachability is arguably one of the most central concepts in temporal graph
algorithmics and has been studied under various aspects, such as path
finding~\cite{wu_efficient_2016,himmel_efficient_2019,xuan_computing_2003,CasteigtsHMZ20},
vertex separation~\cite{kempe2002connectivity,FluschnikMNRZ20,ZschocheFMN20},
finding spanning subgraphs~\cite{CasteigtsPS19,AxiotisF16}, temporal graph
exploration~\cite{ErlebachS18,Erlebach0K15,ErlebachS20,BodlaenderZ19,ErlebachKLSS19,AMS19},
and others~\cite{mertzios2019temporal,AkridaMNRSZ19,BMNR20,HMNR20}.

Perhaps the most prominent application of temporal graph reachability is currently epidemiology,
dealing with effective prevention or containment of disease spreading \cite{covid}.
Here, minimizing the reachability of vertices in a temporal graph by
manipulating the temporal graph corresponds to minimizing the spread of
an infection in various networks by some countermeasures.
Application instances for this scenario may be
drawn from physical contacts \cite{ferretti2020quantifying,eames2003contact} or
airline flights \cite{brockmann2013hidden,colizza2006role}, but also 
social networks \cite{goffman1964generalization,daley1964epidemics}, 
cattle movements \cite{mitchell2005characteristics}, or
computer networks \cite{pastor2001epidemic-big}.

\citet{EnrightMMZ19%
} studied the problem of deleting $k$
time-edges%
\footnote{That is, an edge at a point in time.}
such that no single vertex can reach more than $r$ other vertices and showed
its \NP-hardness and \Wone-hardness for the parameter $k$, even in very restricted
settings.
Here, we shift the focus to a set of multiple given sources, thus studying the
following problem, which has not been considered for computational complexity
analysis yet (to the best of our knowledge).
\problemdef{\textsc{\TemporalDeletingLong (\TemporalDeleting)}}{A temporal graph~$\TG$, a set of
sources vertices $S$, and integers $k, r$.}
{Can we delete at most $k$ time-edges s.t.\ at most $r$ vertices
are reachable from~$S$?}

Imaginably, removing  edges or vertices is not the most \emph{infrastructure friendly} approach  to restrict reachability.
To address this, other operations have been studied.
\citet{enright2021assigning} considered restricting the reachability by 
just changing the relative order in which edges are active.
\citet{DeligkasP20} considered restricting the reachability by a merging operation of consecutive edge sets of the temporal graph
and by a delay operation of time-edges by $\delta$ time steps, i.e., moving a time-edge from time $t$ to $t+\delta$.
In particular, they introduced a delay variant of \TemporalDeleting. 
\problemdef{\textsc{\TemporalDelayingLong (\TemporalDelaying)}}
{A temporal graph~$\TG$, a set of sources vertices $S\subseteq V$, and integers $k,r, \delta$.}
{Can we delay at most $k$ time-edges by $\delta$ s.t.\ at most $r$ vertices are
reachable from~$S$?} This is the central problem studied in this paper.
Throughout the whole paper we assume for all instances of \TemporalDeleting{} and \TemporalDelaying{} that $0 < |S| \leq r$.
We remark that technically \citet{DeligkasP20} formulate the problem slightly differently, allowing delays of \emph{up to}~$\delta$ to appear.
However, a simple argument can be given to see that this distinction is not significant:
Clearly, delaying a time-edge reduces the number of reachable vertices only if the undelayed time-edge could be reached from some source~$s \in S$.
But when this is the case, increasing the delay of that time-edge can never increase the set of vertices reachable from~$S$.
The reason is that, while increasing the delay might enable some source~$s' \in S \setminus \{s\}$ to reach a vertex~$v$,
that vertex~$v$ would be reached from~$s$ in any case.

\citet{DeligkasP20} showed that \TemporalDelaying{} is \NP-hard and \Wone-hard when parameterized by $k$, 
even if underlying graph has lifetime $\lifetime=2$.
A close look into the proof reveals that this also holds for \TemporalDeleting{}.

\subparagraph{Our contribution.}
We study how \TemporalDeleting{} and \TemporalDelaying{} relate to each other.
We show that both problems are polynomial-time solvable on trees.
Moreover, there is an intermediate reduction from \TemporalDeleting{} to \TemporalDelaying{} 
indicating that \TemporalDelaying{} seems generally harder than \TemporalDeleting{}.
However, surprisingly, this is no longer true when we parameterize the problems by
the number $r$ of reachable vertices.
Here, we develop a max-flow-based branching strategy and obtain fixed-parameter tractability for \TemporalDelaying{} while \TemporalDeleting{} remains \Wone-hard.
This makes \TemporalDelaying{} particularly interesting for applications where the number of reachable vertices should be very small,
e.g.\ when trying to contain the spread of dangerous diseases.

\section{Preliminaries}
We define $\NN$ as the positive natural numbers, $[a,b] := \{i \in \mathbb Z \mid a\leq i \leq b\}$, and $[n]:=[1,n]$.
For a function $f \colon V \rightarrow \mathbb Z$ and subset $X \subseteq V$ we denote by $f(X)$ the sum $\sum_{x\in V} f(x)$.
We use standard notation from graph theory \cite{Die16}.
We say for a (directed) graph $G$ that $G - X := G[V(G)\setminus X]$ is the induced subgraph of $G$ 
when the vertices in $X$ are removed,
and $G \setminus Y := (V(G),E(G)\setminus Y)$ is the subgraph when the edges in $Y$ are removed, 
where $X$ is a vertex set and $Y$ is an edge set.
For any predicate~$P$, the Iverson bracket~$[P]$ is 1 if $P$ is true and 0 otherwise.

\subparagraph{Parameterized complexity.}
Let~$\Sigma$ denote a
finite alphabet.
A parameterized problem~$L\subseteq \{(x,k)\in \Sigma^*\times \mathbb N \cup \{0\} \}$ is a subset of all instances~$(x,k)$ from~$\Sigma^*\times \mathbb N \cup \{0\}$,
where~$k$ denotes the \emph{parameter}.
A parameterized problem~$L$ is in 
\FPT (is \emph{fixed-parameter tractable}) if there is an algorithm that decides
every instance~$(x,k)$ for~$L$ in~$f(k)\cdot |x|^{O(1)}$ time,
where~$f$ is any computable function only depending on the parameter.
If a parameterized problem $L$ is W[1]-hard, then it is presumably not
fixed-parameter tractable.
We refer to \citet{DF13} for details.

\subparagraph{Temporal graphs.}
A \emph{temporal graph} $\TG$ consists of a set of vertices~$V$ (or $V(\TG)$),
and a sequence of edge sets~$(E_i)_{i \in [\tau]}$
where each $E_i$ is a set of unordered pairs from $V$.
The number~$\tau$ is called the \emph{lifetime} of $\TG$.
The elements of $\TE(\TG) := \bigcup_{i \in [\tau]} E_i \times \{i\}$ are called the \emph{time-edges} of~$\TG$.
Furthermore $\TG$ has a \emph{traversal time} function $\duration: \TE(\TG) \to \NN$ specifying the time it takes to traverse each time-edge.
The temporal graph $\TG$ is then written as the tuple~$(V, (E_i)_{i \in [\tau]}, \duration)$.
Often we assume $\duration$ to be the constant function $\duration \equiv 1$
and then simply write $\TGcompact$.
The \emph{underlying graph} of $\TG$ is the graph $(V,\bigcup_{i=1}^\lifetime E_i)$.
For a time-edge set $Y$ and temporal graph $\TG$, 
we denote by $\TG \setminus Y$ the temporal graph 
where~$V(\TG\setminus Y)= V(\TG)$ 
and~$\TE(\TG\setminus Y) = \TE(\TG) \setminus Y$.
A \emph{temporal $s$-$z$-path} in~$\TG$ 
is a sequence of time-edges $P=(e_i = (\{v_{i-1},v_i\},t_i))_{i=1}^m$ where
\begin{enumerate}
\item $v_0=s$ and $v_m=z$,
\item the sequence of edges $(\{v_{i-1},v_i\})_{i=1}^m$ forms an $s$-$z$-path in the underlying graph of $\TG$, and 
\item $t_{i+1} \geq t_{i} + \duration(e_i)$ for all $i \in [m-1]$.
\end{enumerate}
The \emph{arrival time} of $P$ is $t_m+\duration(e_m)$.
The set of vertices of $P$ is denoted by $V(P) = \{ v_i \mid 0 \leq i \leq m\}$.
A vertex $w$ is \emph{reachable} from $v$ in $\TG$ (at time~$t$) if there exists a temporal $v$-$w$-path in $\TG$ (with arrival time at most~$t$).
In particular, every vertex reaches itself via a trivial path.
Furthermore, $w$ is reachable from $S \subseteq V$ if there is a temporal $s$-$w$-path for some $s \in S$,
and the set of all vertices reachable from~$S$ is denoted the \emph{reachable set} $\reach_\TG(S)$.
We drop the index $\TG$ if it is clear form the context.
\emph{Delaying} a time-edge $(\{v,w\}, t)$ by $\delta$ refers to replacing it with the time-edge $(\{v, w\}, t+\delta)$.
For a temporal graph $\TG$ and a time-edge set $X \subseteq \TE(\TG)$ we denote by $\TG \delay X$ the temporal graph $\TG$ where the time-edges in $X$ are delayed by $\delta$.

\subparagraph{Preliminary observations.}
We present an intermediate polynomial-time reduction from the \TemporalDeleting{} to \TemporalDelaying{}.
\begin{lemma}
		\label{lem:delete-to-delay}
		Given an instance $I=(\TGcompact,S,k,r)$ of \TemporalDeleting{}, we can compute in linear time,
		an instance $J=(\TG'=(V',(E_i')_{i=1}^{6\lifetime+1}),S',k,r',\delta=3\lifetime)$ of \TemporalDelaying{} such that
		the feedback vertex number\footnote{That is, 
		the minimum number of vertices needed to hit all cycles in an undirected graph.}
				of the underlying graph of~$\TG$ and~$\TG'$ is the same, and
		$I$ is a \yes-instance if and only if $J$ is a \yes-instance.
\end{lemma}
\newcommand{\ev}[1]{\tn{e}_{#1}}
\newcommand{\sv}[1]{\tn{s}_{#1}}
\newcommand{\Ve}{V_\tn{e}}
\newcommand{\Vs}{V_\tn{s}}
\begin{proof}
		We construct $\TG' = (V',(E_i')_{i=1}^{3\lifetime+\delta+1})$ in the following way.
		Set
		\begin{align*}
				\Ve{} &:= \{ \ev{uv}, \ev{vu} \mid (\{ v,u\}, t) \in \TE(\TG) \}\}, &
				\Vs{} &:= \{ \sv{vu} \mid \ev{vu} \in \Ve{}\}, \\
				V' 	&:= V \cup \Ve{} \cup \Vs{}, \text{ and} &
				S' 	&:= S \cup \Vs{} \,.
		\end{align*}
		Begin with $E_i'=\emptyset$ for all $i \in [3\lifetime]$.
		Then, add for each time-edge $(\{v,u\},t) \in \TE(\TG)$, 
		the time-edges $(\{v, \ev{vu}\},3t-2),(\{v,\ev{vu}\},3t),
		(\{\ev{vu}, \ev{uv}\},3t-1),
		(\{u, \ev{uv}\},3t-2)$, and $(\{u, \ev{uv}\},3t)$ to $\TG'$.
		Afterwards, add the time-edge $(\{\sv{vu}, \ev{vu}\},6\lifetime+1=3\lifetime+\delta+1)$ for each $ \ev{vu} \in \Ve{}$.
		Finally we set $r' := r + |V' \setminus V|$.
		Since we add for each time-edge of $\TG$ a constant number of vertices and time-edges to $\TG'$,
		we have that $|\TG'| \in O(|\TG|)$ and $\TG'$ can be computed in linear time.
		Moreover, the underlying graph~$G'$ of $\TG'$ is obtained from the underlying graph~$G$ of $\TG$ by subdividing edges and adding leaves,
		thus $G$ and $G'$ have the same feedback vertex number.
		It remains to prove that the two instances are equivalent.
		
		\textbf{($\Rightarrow$):} Let $X$ be a solution for $I$.
		Then, set $X' := \{ (\{\ev{vu}, \ev{uv}\},3t-1) \mid (\{v,u\},t) \in X \}$.
		Pick $s \in S$ and $v \in V$ arbitrary.
		Note that for each temporal $s$-$v$-path $P'$ in $\TG'$,
		we can construct a temporal $s$-$v$-path $P$ in $\TG$
		which uses a time-edge $(\{u,w\},t)$ if and only if $P'$ uses the time-edge $(\{\ev{uw}, \ev{wu}\}, 3t-1)$.
		Furthermore, any temporal $s$-$v$-path in $\TG'\delay X'$ cannot use any delayed time-edge.
		As $s$~and~$v$ are arbitrary, this proves $\reach_{\TG' \delay X'}(S') \cap V \subseteq \reach_{\TG \setminus X}(S)$.
		Thus, at most $r + \abs{V' \setminus V} = r'$ vertices are reachable from $S'$ in $\TG' \delay X'$,
		thus $J$ is a \yes-instance.

		\textbf{($\Leftarrow$):} Let $X'$ be a solution for $J$.
		Begin by observing that delaying any time-edges between $\Vs{}$ and $\Ve{}$ has no effect.
		Consequently, $\reach_{\TG' \delay X'}(\Vs{}) = \Vs{} \cup \Ve{}$ for every possible choice of $X'$.
		Therefore $X'$ is a valid solution if and only if $\abs{\reach_{\TG' \delay X'}(S) \cap V} \leq r' - \abs{\Vs{} \cup \Ve{}} = r$,
		i.e., we only need to study the reachability between vertices in~$V$.
		Because of this, delaying a time-edge connecting two vertices of~$\Ve{}$ has the same effect as deleting that time-edge.
		Next, observe that instead of delaying some time-edge $(\{v, \ev{vu}\}, t)$ which connects vertices of $V$ and $\Ve{}$,
		the same or better reduction of reachability is achieved by instead delaying $(\{\ev{vu}, \ev{uv}\}, t')$,
		with $t' \in \{t - 1, t+1\}$ chosen appropriately.
		Due to this, we may assume without loss of generality that 
		$X' \subseteq \{ (\{\ev{vu}, \ev{uv}\},3t-1) \mid (\{v,u\}, t) \in \TE(\TG)\}$.

		Set $X := \{ (\{v,u\},t) \mid (\{\ev{uv},\ev{vu}\},3t-1) \in X' \}$.
		Let $s \in S$ and $v \in V$ be arbitrary.
		Note that for each temporal $s$-$v$-path $P$ in $\TG$,
		we can construct a temporal $s$-$v$-path $P'$ in $\TG'$ as above.
		Thus $\reach_{\TG \setminus X}(S) \subseteq \reach_{\TG' \delay X'}(S) \cap V$.
		Since we already observed that $\abs{\reach_{\TG' \delay X'}(S') \cap V} \leq  r$,
		we conclude that $I$~is a \yes-instance.
\end{proof}
Note that the reduction in \cref{lem:delete-to-delay} preserves the size~$k$ of the solution and the feedback vertex number of the underlying graph.
We remark that, in exchange for dropping the latter property, one can modify the reduction to instead have $\abs{S'} = \abs{S}$, by simply adding time-edges from $S$ to $\Vs{}$ before all other time-edges.
However, in any case the size $r'$ of the reachable set in $J$ is unbounded in terms of the size $r$ of the reachable set in $I$. 
Unless $\FPT=\Wone$, this is unavoidable as we learn in the next section.

\section{Parameterized by the Reachable Set Size}
In this section the study \TemporalDelaying{} and \TemporalDeleting{} parameterized by the reachable set size $r$.
In particular, our main result in this section is the fixed-parameter tractability of \TemporalDelaying{} parameterized by $r$.
This is in stark contrast to the \Wone-hardness of \TemporalDeleting{} parameterized by $r$ which we show first.
\begin{theorem}
		\label{thm:del-r-whard}
		\TemporalDeleting{} parameterized by $r$ is \Wone-hard, even if $\lifetime=2$.
\end{theorem}
\begin{proof}
	We present a parameterized reduction from the \Wone-hard \cite{DF13} \textsc{Clique} problem parameterized by $\ell$,
	where given a graph $H=(U, F)$ 
	we are asked whether $H$ contains a clique of size $\ell$. 

	Let $H=(U, F)$ be a graph, where $|F|=m$.
	We construct an instance $I=(\TG,\{s\},k= m-{\ell \choose 2},r=1+\ell+{\ell \choose 2})$ of \TemporalDeleting, where
	$\TG := (V,(E_i)_{i\in[2]})$~is the temporal graph given by
	\begin{align*}
			V &\coloneqq U \cup \{s\}
			\cup \left\{ e_f  \mvert  f \in F \right\}, & &\\
			E_1 &\coloneqq \left\{ \{s,e_f\}  
			\mvert f \in F \right\}, \quad \text{ and }\quad
					E_2 \coloneqq \left\{ \{e_{\{u,v\}},u\},\{e_{\{u,v\}},v\} \mvert \{u,v\} \in F \right\}.
	\end{align*}
	Note that $I$ can be constructed in polynomial time.

	\textbf{($\Rightarrow$):} Let $C=(V',F')$ be a clique of size $\ell$ in $H$.
	We set $X := \{ (\{s, e_f\},1) \mid f \in F\setminus F'\}$.
	Note that $|X| \leq k$ and 
	that for each edge $f \in F$ we can reach $e_f$ from $s$ 
	if and only if $f \in F'$.
	Hence, by the construction of $\TG$, a vertex $u \in U$ is reachable from $s$ in $\TG \setminus X$ if and only if $u \in V'$.
	Hence, we can reach $1 + {\ell \choose 2} +\ell$ many vertices from $s$ in $\TG \setminus X$.
	Thus, $I$ is a $\yes$-instance.

	\textbf{($\Leftarrow$):} Let $X \subseteq \TE(\TG)$ be a solution for $I$.
	Without loss of generality, we can assume that $X$ does not contain a 
	time-edge~$(\{e_f,u\},2)$, because it can be replaced by $(\{e_f,s\},1)$.
	Observe that at least ${\ell \choose 2}$ vertices from $\left\{ e_f  \mvert  f \in F \right\}$ are reachable from $s$ in $\TG \setminus X$.
	Since $r=1 + {\ell\choose 2} + \ell$, we can reach from $s$ 
	at most $\ell$ vertices from $U$.
	Hence, $U \cap R_{\TG\setminus X}(\{s\})$ must form a clique of size $\ell$ in $H$.
\end{proof}
Due to \cref{thm:del-r-whard}, we know that 
there is presumably no $f(r)\cdot |\TG|^{O(1)}$-time algorithm to decide whether 
we can keep the reachable set of a vertex $s$ of~$\TG$ small (at most $r$ vertices), by deleting at most $k$ time-edges.
However, this changes when we delay (instead of deleting) at most $k$ edges.
Formally, we show the following.
\begin{theorem}
		\label{thm:fpt-for-r}
		\TemporalDelaying{} is solvable in $O(r!\cdot k \cdot |\TG|)$ time.
\end{theorem}
The proof of \cref{thm:fpt-for-r} is structured as follows.

\begin{description}
		\item[Step 1 (reduction to slowing):] We reduce \TemporalDelaying{} to an auxiliary problem
				which we call \TemporalSlowing. Here, instead of \emph{delaying} a time-edge 
				(moving it $\delta$~layers forward in time) we \emph{slow} it, \ie increase the time required to traverse it by~$\delta$.
		\item[Step 2 (flow-based techniques):] Our new target now is a fixed-parameter algorithm
				for \TemporalSlowing.
				Since we do not aim to preserve a specific temporal graph class, we simplify the input by replacing $S$ with a single-source $s$.
				Then we transform the temporal graph~$\TG$ into a (non-temporal) directed graph~$D$
				in which the deletion of an edge corresponds to slowing a temporal edge in $\TG$.
				Using this, we derive a max-flow-based polynomial-time algorithm which checks whether 
				the source~$s$ can be prevented from reaching any vertices outside of a given set~$R$
				by slowing at most~$k$ time-edges in $\TG$.
		\item[Step 3 (resulting search-tree):] We are aiming for a search-tree algorithm for \TemporalSlowing{}.
				Let $R$ be a set of vertices and suppose that our max-flow-based algorithm failed to prevent $s$ from reaching any vertices outside of~$R$.
				Now, if there exists a solution for the given instance of \TemporalSlowing{},
				then we can identify less than $\abs{R}$ vertices such that at least one of them will be always reached from~$s$.
				We can then try adding each of them to~$R$, gradually building a search-tree to find the solution.
\end{description}

\noindent
Henceforth the details follow.
Instead of solving \TemporalDelaying{} directly, we reduce it to an auxiliary problem introduced next.
Let $\TGtraversal$ be a temporal graph.
\emph{Slowing} a time-edge $(\{v, w\}, t)$ by $\delta$ refers to increasing $\duration((\{v, w\}, t))$ by $\delta$.
We define $\TG \slow X := (V, (E_i)_{i \in [\tau]}, \gamma')$
where $\gamma'(e) := \gamma(e) + \delta\cdot [e \in X]$.
Our auxiliary problem is the following.
\problemdef{\TemporalSlowingLong (\TemporalSlowing)}
{A temporal graph~$\TGtraversal$, a set of
sources $S\subseteq V$, and integers $k,r,\delta$.}
{Is there a time-edge set $X \subseteq \TE(\TG)$ of size at most $k$ such that $\abs{\reach_{\TG \slow X}(S)} \leq r$?}

By the following, solving an instance of \TemporalSlowing{} also solves \TemporalDelaying{}.
\begin{lemma}
		\label{lem:slowing-is-enough-delay}
		An instance $I = (\TGtraversal,S,k,r,\delta)$ of \TemporalDelaying{} is a \yes-instance
		if and only if $J = (\TG,S,k,r,\delta)$ is a \yes-instance of \TemporalSlowing{}.
\end{lemma}
\begin{proof}

		\textbf{($\Rightarrow$):} Let $X \subseteq \TE(\TG)$ be a solution for $I$.
		Note that for every temporal $v$-$w$-path $P$ in $\TG \slow X$,
		we can construct a temporal $v$-$w$-path $P'$ in $\TG \delay X$
		by replacing a time-edge $(e,t) \in P$ with $(e,t+\delta)$ whenever $(e,t) \in X$.
		Hence, the reachable set of $s\in S$ in $\TG \delay X$ is a superset 
		of the reachable set of $s$ in $\TG \slow X$. 
		Thus, $J$ is a \yes-instance.

		\textbf{($\Leftarrow$):} Let $X \subseteq \TE(\TG)$ be an inclusion-minimal solution for $J$.
	  	We claim that every vertex reachable from~$S$ in $\TG \delay X$ by some time~$t$
	  	is also reachable from~$S$ in $\TG \slow X$ until time~$t$.
	  	Suppose for contradiction that the claim does not hold true for some vertex~$z$ and let $t$ be the time $S$~reaches~$z$ in $\TG \delay X$.
	  	We may assume $z$ to be chosen to minimize~$t$.
	  	Clearly $t > 0$, \ie $z \notin S$.
	  	Let $P$ be a temporal $s$-$z$-path in $\TG \delay X$ with arrival time~$t$ and $s \in S$.
	  	Let $u$ be the penultimate vertex of $P$ and
	  	$(\{u,z\}, t')$ the last time-edge of~$P$.
	  	By minimality of $t$, $u$ must be reachable from~$S$ by time~$t'$ also in $\TG \slow X$.
	  	Since all time-edges of $\TE(\TG) \setminus X$ appear in $\TG \delay X$ and $\TG \slow X$ with identical traversal times,
	  	the last time-edge $(\{u,z\},t')$ of~$P$ must be in $\TE(\TG \delay X) \setminus (\TE(\TG) \setminus X)$.
	  	Thus $(\{u,z\},t'-\delta) \in X$.
		By minimality of~$X$, there must be a source~$s' \in S$ and a temporal $s'$-$u$-path $P'$ in $\TG \slow X$ reaching either $u$ or $z$ at time~$t'-\delta$.
		If $P'$ reaches~$z$, then this is clearly a contradiction.
		But if $P'$ reaches~$u$, then appending $(\{u, z\}, t-\delta)$ to $P'$ produces a temporal $s'$-$z$-path in $\TG \slow X$ 	arriving at time~$t$,
		thus also a contradiction.
\end{proof}

In the reminder of this section, we show that \TemporalSlowing{} is fixed-parameter tractable, 
when parameterized by $r$.
Formally, we aim for the following theorem, which in turn clearly implies \cref{thm:fpt-for-r} by the means of \cref{lem:slowing-is-enough-delay}.
\begin{theorem}
		\label{thm:slowing-fpt-r}
		\TemporalSlowing{} can be solved in $O(r!\cdot k \cdot |\TG|)$ time.
\end{theorem}
The remainder of this section is dedicated to proving \cref{thm:slowing-fpt-r}.
The advantage of considering \TemporalSlowing{} instead of \TemporalDelaying{} is that we do not have to deal with new time-edges 
appearing due to the delay operation. 
This allows us to translate the reachability of a temporal graph to a (non-temporal) directed graph specially tailored to \TemporalSlowing{}.
In particular, the removal of some edges in the directed graph 
corresponds to slowing the corresponding time-edges by~$\delta$ in the temporal graph.
Before giving the details of the construction, we first reduce to the case where $S$ is a singleton.

\begin{lemma}
		\label{lem:single-source}
		Given an instance $I = (\TGcompact,S,k,r,\delta)$
		of \TemporalSlowing{},
		we can construct in linear time a instance $J = (\TG',\{s\},k,r+1,\delta)$ of \TemporalSlowing{} such that
		$I$ is a \yes-instance if and only if $J$ is a \yes-instance.
\end{lemma}
\begin{proof}
		We set $\TG' := (V \cup \{s\},(E'_i)_{i\in[\lifetime+\delta+1]})$ 
		where $s$~is a new vertex, 
		$E'_1 := \{\{s, s'\} \mid s' \in S \}$,
		$E'_i := \emptyset$ for all $i \in [\delta+1]\setminus\{1\}$, and
		$E'_{i+\delta+1} := E_i$ for all $i \in [\lifetime]$.
		Observe that slowing an edge in~$E_1'$ has no effect.
		Thus, $I$ is a \yes-instance if and only if $J$ is a \yes-instance.
		Clearly, $J$ can be computed in linear time.
\end{proof}

A \emph{network} $(D,c)$ consists of a directed graph $D=(V,A)$ and \emph{edge capacities} $c \colon A\rightarrow \mathbb N_0 \cup \{\infty\}$.
A function $f \colon A \rightarrow \mathbb N \cup \{0\}$ is an \emph{$s$-$z$-flow} for two distinct vertices $s,z \in V$ if
\begin{itemize}
		\item $\forall e \in A \colon f(e) \leq c(e)$ and
		\item $\forall v \in V \setminus \{s,z\} \colon \sum_{(u,v) \in A} f((u,v)) = \sum_{(v,u) \in A} f((v,u))$.
\end{itemize}
The \emph{value} of $f$ is denoted by $ |f| \coloneqq \sum_{(s,v)\in A} f((s,v))$.
An arc set $C \subseteq A$ is an \emph{$s$-$z$-cut} of a network $((V,A),c)$ if $s,z \in V$ and there is no $s$-$z$-path in $(V, A \setminus C)$.
The \emph{capacity} of the $s$-$z$-cut $C$ is $c(C) := \sum_{e\in C} c(e)$.

Let $\TGtraversal$ be a temporal graph.
We define the \emph{temporal neighborhood} of a vertex $v\in V$ at time point $t \in [\lifetime]$ 
as the set~$N_\TG(v,t) := \bigcup_{i=t}^{\lifetime} N_{(V,E_i)}(v)$ containing all neighbors of~$v$ in the layers $t$ through $\lifetime$.

For any $s \in R \subseteq V$ and $\delta \in \NN$,
we define the flow network
$\FlowNet(\TG, s, R, \delta) := (D, c)$ where 
$D = (V', A)$ is the directed graph defined by
\newcommand{\flowsource}{s^0}
\newcommand{\flowsink}{\tn{z}}
\begin{align}
	V' :={}& \{ \flowsource{}, \flowsink{} \}
		\cup \left\{
			\begin{array}{l}
				e_1,e_2, v^t,w^t,v^{t+\gamma(e)},w^{t+\gamma(e)}, \\
				v^{t+\gamma(e)+\delta} ,w^{t+\gamma(e)+\delta}
			\end{array}
		\mvert v, w \in R \text{ and } e=(\{v,w\},t) \in \TE(\TG)  \right\} 
	\nonumber
\end{align}
and
\begin{align}
	A :={}&  \left\{ (v^t,v^{t'}) \mvert  v^t \in V', t' = \min\{i \mid i > t \text{ and } v^i \in V'\} \neq \infty\right\}
		\label{eq:A3} \\
		& \cup \left\{
			\begin{array}{l}
				(v^t,e_1),(w^t,e_1),
				(e_1,e_2),\\
				(e_2,v^{t+\gamma(e)}),(e_2,w^{t+\gamma(e)}),\\
				(v^{t},w^{t+\gamma(e)+\delta}),
				(w^{t},v^{t+\gamma(e)+\delta})
			\end{array}
			\mvert v, w \in R \text{ and } e=(\{v,w\},t) \in \TE(\TG)  \right\}
		\label{eq:A4}\\
		&\cup \left\{ (v^t,\flowsink{})\mvert
			v \in R,  t = \max \left\{i \mid N_{\TG}(v, i) \nsubseteq R\right\} \neq -\infty
			\right\}.
		\label{eq:A5}
\end{align}
and we set $c((e_1,e_2))=1$ for all $e \in \TE(\TG)$
and $c(a) = \infty$ for all other $a \in A$.
Consider \cref{fig:digraph} for an illustration.
\begin{SCfigure}[\sidecaptionrelwidth][t!]

		\begin{tikzpicture}
				\node[vertex,label=$v$] (v) at (0,0) {};
				\node[vertex,label=$w$] (u) at (2,0) {};
				\draw[edge] (v) to node[fill=white] {$3$} (u);
				\draw[edge,gray,dotted] (v) to (-0.3,-0);
				\draw[edge,gray,dotted] (v) to (-0.3,-0.25);
				\draw[edge,gray,dotted] (v) to (-0.3,0.25);

				\draw[edge,gray,dotted] (u) to (2.3,-0);
				\draw[edge,gray,dotted] (u) to (2.3,-0.25);
				\draw[edge,gray,dotted] (u) to (2.3,0.25);

				\begin{scope}[xshift=4cm,yshift=1cm]
				\node[vertex,label=left:$v^3$] (v3) at (0,0) {};
				\node[vertex,label=left:$v^4$] (v4) at (0,-1.5) {};
				\node[vertex,label=left:$v^5$] (v5) at (0,-2) {};
				\draw[thick,->] (v3) to (v4);
				\draw[thick,->] (v4) to (v5);

				\node[vertex,label=right:$w^3$] (u3) at (2,0) {};
				\node[vertex,label=right:$w^4$] (u4) at (2,-1.5) {};
				\node[vertex,label=right:$w^5$] (u5) at (2,-2) {};
				\draw[thick,->] (u3) to (u4);
				\draw[thick,->] (u4) to (u5);

				\draw[edge,gray,dotted] (v3) to (-0.2,-0);
				\draw[edge,gray,dotted] (v3) to (-0.2,-0.25);
				\draw[edge,gray,dotted] (v3) to (-0.2,0.25);

				\draw[edge,gray,dotted] (v4) to (-0.2,-1.5);
				\draw[edge,gray,dotted] (v4) to (-0.2,-1.75);
				\draw[edge,gray,dotted] (v4) to (-0.2,-1.25);

				\draw[edge,gray,dotted] (v5) to (-0.2,-2);
				\draw[edge,gray,dotted] (v5) to (-0.2,-2.25);
				\draw[edge,gray,dotted] (v5) to (-0.2,-1.75);

				\draw[edge,gray,dotted] (u3) to (2.2,-0);
				\draw[edge,gray,dotted] (u3) to (2.2,-0.25);
				\draw[edge,gray,dotted] (u3) to (2.2,0.25);

				\draw[edge,gray,dotted] (u4) to (2.2,-1.5);
				\draw[edge,gray,dotted] (u4) to (2.2,-1.75);
				\draw[edge,gray,dotted] (u4) to (2.2,-1.25);

				\draw[edge,gray,dotted] (u5) to (2.2,-2);
				\draw[edge,gray,dotted] (u5) to (2.2,-2.25);
				\draw[edge,gray,dotted] (u5) to (2.2,-1.75);

				\node[vertex,label=left:$e_1$] (e1) at (1,-0.5) {};
				\node[vertex,label=left:$e_2$] (e2) at (1,-1.25) {};
				\draw[thick,dashed,->] (e1) to (e2);
				\draw[thick,->] (v3) to (e1);
				\draw[thick,->] (e2) to[bend right](u4);
				\draw[thick,->] (e2) to[bend left](v4);
				\draw[thick,->] (u3) to (e1);

				\draw[thick,->] (v3) to[bend right=45] (u5);
				\draw[thick,->] (u3) to[bend left=45] (v5);
				\end{scope}
		\end{tikzpicture}
		\caption{\emph{Left}: An excerpt of a temporal graph $\TG$ containing the time-edge $e=(\{v,w\},3)$ and $\gamma(e)=1$.
				\emph{Right}: An excerpt of the flow network $\FlowNet(\TG, s, R, \delta=1	)$ showing the corresponding part for $e$, 
	where solid arc have capacity~$\infty$ and the dashed arc has capacity~$1$.}
	\label{fig:digraph}
\end{SCfigure}
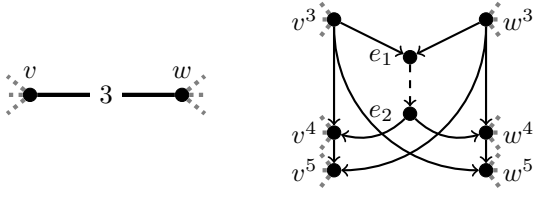

\begin{lemma}
		\label{lem:digraph-flow}
		For any given $\TGtraversal{}$, $s \in R \subseteq V$, and $k, \delta \in \NN$,
		one can test in $O(k \cdot \abs{\TG})$~time whether $\FlowNet(\TG, s, R, \delta)$ has a $\flowsource{}$-$\flowsink{}$-flow of value at least $k+1$
		and compute such a flow or a maximum flow otherwise.
\end{lemma}
\begin{proof}
		Note, that the flow network $\FlowNet(\TG, s, R, \delta)$ can be computed in $O(|\TG|)$ time
		by iterating over $\TE(\TG)$ first forward and then backwards once.
		Then we can compute a flow of value~$k+1$ or of maximum value, whichever is smaller,
		by running at most $k+1$ rounds of the Ford-Fulkerson algorithm \cite{ford_fulkerson_1956}.
		This gives a overall running time of $O(k \cdot \abs{\TG})$ time.
\end{proof}

\begin{lemma}\label{lem:path-equivalence}
	Let $\TGtraversal{}$, $s \in R \subseteq V$, $\delta \in \NN$,
	and $((V', A), c) := \FlowNet(\TG, s, R, \delta)$.
	Let $X \subseteq \TE(\TG)$ and $C := \{(e_1, e_2) \in A \mid e \in X\}$.
	For any $x^t \in V'$, there is a $\flowsource$-$x^t$-path in $(V', A \setminus C)$
	if and only if $\TG \slow X$ contains a temporal $s$-$x$-path with arrival time at most~$t$.
\end{lemma}
\begin{proof}
	Let $\gamma'$ be the traversal time function of $\TG \slow X$, \ie $\gamma'(e) := \gamma(e)+\delta\cdot [e \in X]$.
	
	\textbf{($\Leftarrow$):} 
		Let $P$~be a temporal $s$-$x$-path in $\TG \slow X$ for some vertex~$x \in V$
		and let $t$~be the arrival time of $P$.
		Then we construct a $\flowsource$-$x^t$-path $\hat{P}$ in $(V', A \setminus C)$ as follows.
		Start with $\hat{P}$ being just the vertex~$s^0$ and perform the following two steps for every time-edge $e = (\{v, w\}, b)$ of~$P$ in order.
		\begin{enumerate}
			\item
				Note that the currently last vertex of~$\hat{P}$ is $v^c$ for some~$c \leq b$.
				As long as~$c < b$, append to~$\hat{P}$ the arc $(v^c, v^d)$ where $d > c$ is chosen minimal (cf. \cref{eq:A3}).
				Afterwards, the currently last vertex of~$\hat{P}$ is~$v^b$.
			\item
				If $e \notin X$, \ie $(e_1, e_2) \notin C$,
				then append to~$\hat{P}$ the arcs $(v^b, e_1)$, $(e_1, e_2)$, and $(e_2, w^{b + \gamma(e)})$ (cf. \cref{eq:A4}).
				Otherwise, if $e \in X$, \ie $\gamma'(e) = \gamma(e) + \delta$,
				then append to~$\hat{P}$ the arc $(v^b, w^{b + \gamma(e) + \delta})$.
				Note that in both cases the new last vertex of $\hat{P}$ is $w^{b + \gamma'(e)}$.
		\end{enumerate}
		
	\textbf{($\Rightarrow$):}
		Let $P$ be a $\flowsource$-$x^t$-path in $(V', A \setminus C)$.
		Then we construct a $s$-$x$-path~$P'$ with arrival time at most~$t$ in $\TG \slow X$ as follows.
		Start with $P'$ being just the vertex $s$ (with arrival time $0$) and repeat the following steps until all arcs of~$P$ have been processed.
		\begin{enumerate}
		\item
			Let $b$ be the arrival time of $P'$ and $v$ the last vertex.
			Note that the last vertex of $P$ is $v^c$ for some $c \geq b$.
			Ignore all arcs of~$P$ up to the last arc containing $v^c$ for some $c \geq b$.
		\item
			If the next three arcs in $P$ are $(v^c, e_1), (e_1, e_2), (e_2, w^{c + \gamma(e)})$
			for some time-edge $e = (\{v, w\}, c) \in \TE(\TG)$,
			then append to~$P'$ that time-edge~$e$.
			Note that, by assumption, $(e_1, e_2) \notin C$, thus $e \notin X$, and thus the arrival time of~$e$ is $c + \gamma'(e) = c + \gamma(e)$.
		\item
			Otherwise the next arc in $P$ must be $(v^c, w^{c + \gamma(e) + \delta})$ for some time-edge $e = (\{v, w\}, c) \in \TE(\TG)$.
			Then append to~$P'$ that time-edge~$e$.
			Note that the arrival time of~$e$ is $c + \gamma'(e) \leq c + \gamma(e) + \delta$.
		\end{enumerate}
\vspace{-\baselineskip}
\end{proof}

We now show that we can use \cref{lem:digraph-flow} to check whether $s$ can be prevented from reaching any vertices outside of~$R$
by slowing at most $k$~time-edges by~$\delta$ each.
\begin{lemma}
		\label{lem:the-good-case}
		For any given $\TGtraversal{}$, $s \in R \subseteq V$, $\delta \in \NN$, and $k \in \NN \cup \{ 0 \}$,
		the maximum $\flowsource{}$-$\flowsink{}$-flow in $\FlowNet(\TG, s, R, \delta)$ has value at most $k$
		if and only if
		there is a set~$X$ of at most~$k$ time-edges such that
		$s$ can not reach any vertices outside of~$R$ in $\TG \slow X$.
\end{lemma}
\begin{proof}
		Write $((V',A),c) := F := \FlowNet(\TG, s, R, \delta)$.
		
		\textbf{($\Rightarrow$):} Let the maximum $\flowsource$-$\flowsink$-flow~$f$ in~$F$ have value at most~$k$.
		Moreover, let $C$ be a $\flowsource$-$\flowsink$-cut of minimum capacity.
		From the max-flow min-cut theorem \cite{maxflowmincut}, we know that $c(C) \leq k$.
		Note that $C \subseteq \{ (e_1,e_2) \in A \mid e \in \TE(\TG) \}$, 
		since all other edges have infinite capacity.
		Hence, $|C| \leq k$.
		Now set $X := \{ e \in \TE(\TG)\mid (e_1,e_2) \in C \}$.
		Assume towards a contradiction that there is a temporal $s$-$x$-path $P$ 
		in $\TG \slow X$ for some $x \in V \setminus R$.
		We may take $P$ to be minimal, thus the penultimate vertex~$y$ of $P$ is contained in $R$.
		By \cref{lem:path-equivalence} there is a $\flowsource$-$y^t$-path $\hat{P}$ in $(V', A \setminus C)$ where $t$~is the time $P$~reaches~$y$.
		The fact that $P$ afterwards proceeds to~$x$ and $\cref{eq:A5}$ in the definition of $\FlowNet$ imply that
		$(V', A \setminus C)$~contains a path from $y^t$ to $\flowsink$.
		This contradicts $C$~being a $\flowsource$-$\flowsink$-cut in $(V', A)$.

		\textbf{($\Leftarrow$):}
		Let $X$ be a time-edge set as assumed.
		By assumption $\reach_{\TG \slow X}(\{s\}) \subseteq R$.
		We claim that $C= \{ (e_1,e_2) \mid e\in X \} \subseteq A$ is a $\flowsource$-$\flowsink$-cut in~$(V', A)$,
		which implies that the maximum value of an $\flowsource$-$\flowsink$-flow in~$F$ is at most $c(C) \leq k$ \cite{maxflowmincut}.
		So suppose towards a contradiction that there is a $\flowsource$-$\flowsink$-path~$P$ in~$(V', A \setminus C)$.
		Let $x^t\in V$ be the penultimate vertex of~$P$.
		Then there is a $s$-$x$-path~$P'$ in $\TG \slow X$ with arrival time at most~$t$ by \cref{lem:path-equivalence}.
		The final arc of~$P$ is $(x^t, \flowsink)$.
		Hence, $(x^t, \flowsink)$ must be contained in \cref{eq:A5}, \ie we can extend $P'$ by some time-edge to end at a vertex in $V \setminus R$.
		This contradicts our assumption $\reach_{\TG \slow X}(\{s\}) \subseteq R$.
\end{proof}

If $\FlowNet(\TG, s, R, \delta)$ contains a $\flowsource$-$\flowsink$-flow of value~$k+1$, then we want to find
a small set~$Y \subseteq V(\TG) \setminus R$ of vertices such that $Y \cap \reach_{\TG \slow X}(s) \neq \emptyset$
for every $X \subseteq \TE(\TG)$ with $\abs{X} \leq k$ and $\abs{\reach_{\TG \slow X}(s)} \leq r$.

\begin{lemma}
		\label{lem:small-choice}
		Let $\TGtraversal{}$, $s \in R \subseteq V$, $\delta, r \in \NN$, and $k \in \NN \cup \{ 0 \}$.
		Assume that $\FlowNet(\TG, s, R, \delta)$ has a $\flowsource{}$-$\flowsink{}$-flow of value~$k+1$.
		We can compute in $O(k \cdot |A|)$ time a set $Y \subseteq V \setminus R$ of size at most $\abs{R}$
		such that $Y \cap \reach_{\TG \slow X}(s) \neq \emptyset$
		holds for every $X \subseteq \TE(\TG)$ with $\abs{X} \leq k$ and $\abs{\reach_{\TG \slow X}(s)} \leq r$.
\end{lemma}
\begin{proof}
		Let $f$ be a $\flowsource$-$\flowsink$-flow of value $k+1$ in $((V',A), c) := \FlowNet(\TG, s, R, \delta)$.
		We may assume $f$ to never use an arc $(v^t, w^{t'})$ whenever $A$~contains some arc $(v^b, \flowsink)$ with $b \geq t$,
		as we could otherwise redirect~$f$ to use that latter arc (of infinite capacity) instead.
		Note that performing this modification can be done in $O(k \cdot |A|)$ time.
		Now set
		\begin{align*}
				H  &:= \left\{ (v,t) \mvert v^t \in V' \text{ and } (v^t, \flowsink) \in A, f((v^t, \flowsink)) > 0\right\}
		\end{align*}
		By \cref{eq:A5}, we have $|H|\leq |R|$ and $N_\TG(v,t)\setminus R \not = \emptyset$ for all $(v,t) \in H$.
		Construct the vertex set $Y$ by picking for each $(v,t) \in H$ one arbitrary vertex from $N_\TG(v,t)\setminus R$.
		Hence $|Y| \leq |H| \leq |R|$ and $Y \cap R = \emptyset$.

		It remains to prove that $\reach_{\TG \slow X}(s) \cap Y \neq \emptyset$,
		with $X \subseteq \TE(\TG)$ being an arbitrary solution for the \TemporalSlowing{}-instance $(\TG,\{s\},k,r,\delta)$.
		Define $C:= \{ (e_1,e_2) \in A \mid e \in X \}$.
		Since $f$ has value $k+1 > c(C)$, there is a $\flowsource$-$v^t$-path $P$ in $(V', A \setminus C)$ 
		where each edge $e\in E(P)$ has $f(e) > 0$ and $(v,t) \in H$.
		By \cref{lem:path-equivalence}, there is a temporal $s$-$v$-path $P'$ in $\TG \slow X$ with arrival time at most $t$.
		Note that through~$P'$, vertex $s$ can reach $v$ as well as all vertices of $N_{\TG \slow X}(v, t) = N_{\TG}(v, t)$.
		Hence, the vertex $u \in Y$ which we picked for $(v,t)\in H$ from~$N_\TG(v,t)\setminus R$
		is in $\reach_{\TG \slow X}(s)$ --- thus the claim is proven.

\end{proof}

Now we are ready to prove of \cref{thm:slowing-fpt-r}. 
The corresponding search-tree algorithm is listed as \cref{algo:search-tree}.
\begin{algorithm2e}[t]
\addtocounter{theorem}{1}
		\KwIn{An instance $I=(\TG,\{s\},k,r,\delta)$ of \TemporalSlowing{}.}
		\KwOut{\yes{} if $I$ is a \yes-instance and otherwise \no.}
		\SetKwProg{Fn}{function}{ is}{end}
		\KwRet{$g(\{s\})$, where}\;
		\Fn{$g(R)$}{
				Compute a $\flowsource$-$\flowsink$-flow $f$ in $\FlowNet(\TG, s, R, \delta)$ by \cref{lem:digraph-flow}\label{line:flow}\;
				\lIf(){$f$ is of value at most $k$}{\KwRet{\yes}\label{line:5}  }
				\lIf{$|R|\geq r$}{ \KwRet{\no}}\label{line:3}
				Compute a set $Y \subseteq V \setminus R$ by \cref{lem:small-choice}\label{line:6}\;
				\ForEach{$v \in Y$}{
						\lIf{$g(R \cup \{v\})=$ \yes}{\KwRet{\yes}\label{line:9}}
				}
				\KwRet{\no}
		}
\caption{Pseudocode of the algorithm behind \cref{thm:slowing-fpt-r}.}
\label{algo:search-tree}
\end{algorithm2e}

\begin{proof}[Proof of \cref{thm:slowing-fpt-r}]
Let $I=(\TG,S,k,r,\delta)$ be the given instance of \TemporalSlowing{}.
By \cref{lem:single-source}, we can assume $S=\{s\}$ after linear time preprocessing.
We now prove that \cref{algo:search-tree} solves $(\TG, \{s\}, k, r, \delta)$ in $\bigO(r^{2r}\cdot k \cdot |\TG|)$ time.

Let $I$ be a \no-instance.
Observe that line~\ref{line:3} ensures that at all times $\abs{R} \leq r$.
Then, by \cref{lem:the-good-case}, $\FlowNet(\TG, s, R, \delta)$
will for all $s \in R \subseteq V$ have a $\flowsource$-$\flowsink$-flow of value $k+1$.
Hence, line~\ref{line:5}, and thus \cref{algo:search-tree} will never return \yes.

Let $I$ be a \yes-instance.
Thus there is a set~$X$ of at most $k$ time-edges such that $\abs{\reach_{\TG \slow X}(s)} \leq r$.
We claim that $g(R')$ returns \yes{}
for all $R'$ with $s \in R' \subseteq \reach_{\TG \slow X}(s)$.
We prove this by reverse induction on $\abs{R'}$.
In the base case where $R = \reach_{\TG \slow X}(s)$, $g(R)$ returns \yes{} by \cref{lem:the-good-case}.
Now assume the claim to hold whenever $\abs{R} = q$
and let $R$ be of size $q-1$ with $s \in R \subseteq \reach_{\TG \slow X}(s)$.
Assume that $\FlowNet(\TG, s, R, k, \delta)$ has a $\flowsource$-$\flowsink$-flow of value $k+1$,
otherwise we are done (by line~\ref{line:5}).
By \cref{lem:small-choice}, 
the set~$Y$ computed in line~\ref{line:6} contains a vertex $u \in \reach_{\TG \slow X}(s) \setminus R$.
Thus, $g(R \cup \{u\})$ returns \yes{} by induction hypothesis.
Hence, by line~\ref{line:9}, $g(R)$ return \yes{}, completing the induction.
In particular, $g(\{s\})$ returns \yes, therefore \cref{algo:search-tree} is correct.

To bound the running time, note that each call~$g(R)$ makes at most $\abs{Y} \leq |R|$ recursive calls by \cref{lem:small-choice}. 
In each of these recursive calls the cardinality of $R$ increases by one until $|R| = r$, 
so the recursion depth is at most~$r$ by line~\ref{line:3}.
Hence, we can observe by an inductive argument that the search tree has $d!$ many nodes at depth $d$, where the root is at depth $1$.
Hence, the search tree has at most $r!$ many leaves and thus $\bigO(r!)$ many nodes in total.
By \cref{lem:digraph-flow} and \cref{lem:small-choice}, lines~\ref{line:flow}~and~\ref{line:6}
take at most $O(k \cdot |\TG|)$ time.
Hence, the overall running time is bounded by $\bigO(r!\cdot k \cdot |\TG|)$.
\end{proof}

\section{A Polynomial-Time Algorithm for Forests} 
In this section we present an algorithm that solves \TemporalDeleting{} and
\TemporalDelaying{} in polynomial time on temporal graphs where the underlying
graph is a tree or a forest.
This is a quite severe yet well-motivated restriction of the
input~\cite{enright2018deleting,enright2021assigning,EnrightMMZ19}, since it could serve as the starting point for FPT-algorithms for ``distance-to-forest''-parameterizations.
\begin{theorem}
		\TemporalDeleting{} and \TemporalDelaying{} are polynomial-time solvable
		if the underlying graph is a forest.
\end{theorem}
Actually we even provide an polynomial-time algorithm for a 
generalized version of \TemporalDelaying{}.
Then, polynomial-time solvability of  \TemporalDeleting{} follows 
from \cref{lem:delete-to-delay}, since it is forest preserving. 
We define the generalized problem as follows:
\newcommand{\TemporalDelayingOnTrees}{\textsc{Weighted} \TemporalDelaying \textsc{on Forests}}
\problemdef{\TemporalDelayingOnTrees}{
		A temporal graph~$\TGtraversal$ whose underlying graph is a forest, 
		a weight function~$w \colon V \rightarrow \mathbb N \cup \{0,\infty\}$,
		a set~$F \subseteq \TE(\TG)$ of undelayable time-edges,
		a set of sources~$S\subseteq V$, 
		and integers $k,r, \delta$.}
{Does there exists a time-edge set $X \subseteq \TE(\TG) \setminus F$ of size at most $k$ such that~$w(\reach_{\TG\delay X}(S)) \leq r$?}

In the remainder of this section, we show how to solve this problem using dynamic programming in polynomial time.
Informally speaking, our dynamic program works as follows.
As a preprocessing step we unfold vertices of large degree, reducing to an equivalent instance of maximum degree~3.
Then we root each underlying tree at an arbitrary vertex and build a dynamic programming table, starting at the leaves.
More precisely, we compute a table entry for each combination of a vertex $v$, a budget $k$, a time step $t$, and a flag
indicating whether $v$~is first reached from a child or from its parent.
This table entry then contains a minimum reachable subset of the subtree rooted at~$v$ that can be achieved by applying~$k$~delay operations to that subtree.
\begin{theorem}
		\label{thm:poly-forests}
		\TemporalDelayingOnTrees{} is polynomial-time solvable if the underlying graph is a forest.
\end{theorem}
\begin{proof}
Assume for now that the underlying graph of $\TG$ is a tree, rooted at an arbitrary leaf.
We denote by $T_v$ the subtree with root $v \in V$.
We use the reaching time $\infty$ to denote ``never''.
By convention, a vertex can reach itself at time~$0$.
Define $\NN^* := \NN \cup \{0,\infty\}$.

We first show how to transform the underlying graph into a binary tree. This
will highly simplify the description of the dynamic programming table.
Replace each vertex $v$ of degree $\deg(v) > 3$ by a path on $\deg(v) - 2$ new vertices,
where each edge of that path is undelayable, appears at each time step and always has traversal time 0.
Distribute the edges formerly incident to $v$ among the new vertices such that each of them has degree~$3$.
Set the weight of the path's first vertex to the weight of $v$, and the weight of all other path vertices to~$0$.
Note that this modification produces an equivalent instance of maximum degree at most $3$,
while increasing the number of vertices only by a constant factor.

We extend the notion of reachability to vertex-time pairs $(s, t) \in S \times [\lifetime]$
by saying that~$(s, t)$ reaches $v \in V$ in $\TG$ if there exists a temporal $s$-$v$-path starting at time~$t$ or later.
For a set $A \subseteq V \times [\lifetime]$, $R_\TG(A)$ is the set of all vertices, reachable from any member of $A$ in $\TG$.
We say a vertex~$v$ is reached \emph{through} another vertex~$w$ if there is a temporal path from a source~$s \in S$ to $v$ that uses~$w$.

Let $v \in V$, $k \in \NN$.
Define $\Tt_{v, k}$ as the set of temporal graphs obtained from $T_v$ by applying up to $k$ delay operations.
Partition $\Tt_{v,k}$ into $\{\Tt_{v,k,t} \mid t \in \NN^*\}$,
where $\Tt_{v,k,t}$ contains those graphs in which $v$ is reached from $S \cap V(T_v)$ exactly at time $t$.
Finally, we set for each $t \in \NN^*$
\begin{align*}
D[v, k, t, \false] &=
\min\left\{ w\bigl(\reach_T(S \cap V(T_v)) \bigr) \mvert T \in \Tt_{v, k, t} \right\} \quad\text{and}\\
D[v, k, t, \true] &= \min\left\{ w\bigl(\reach_T(S \cap V(T_v)) \cup \reach_T(\{(v, t)\})  \bigr) \mvert T \in \bigcup_{t' \geq t}\Tt_{v, k, t'} \right\} .
\end{align*}
where the minimum of an empty set is $\infty$ by convention.
It is convenient to also define these entries as $\infty$ whenever $k < 0$.
Roughly speaking, $D[v,k,t,\iota]$ contains the minimal weight reached in $T_v$ under the assumption that up to~$k$ delay operations are applied to $T_v$, that $v$ is first reached at time~$t$, and that
\begin{itemize}
		\item $v$ is reached by a source in $S \cap V(T_v)$ at time~$t$ if $\iota = \false$, 
		\item $v$ is reached by a source in $S \setminus V(T_v)$ at time~$t$ if $\iota = \true$.
\end{itemize}
Note that $v$ might be reached simultaneously from $S \cap V(T_v)$ and $S \setminus V(T_v)$.

We next show how to compute $D[v,k,t,\iota]$ recursively, starting at the leaf
vertices.
Observe that $D[v,k,t,\iota] = \infty$ whenever $v \in S$ is a source and $t > 0$.
Thus, this case shall be excluded in the following.

\subparagraph*{If $v$ has no children.}
If $\iota = \false$ and $v \notin S$ and $t < \infty$, then $D[v,k,t,\false] = \infty$
as there is no way that $v$ can be reached from a source in $S \cap V(T_v) = \emptyset$.
Otherwise,
\[
	D[v,k,t,\iota] = w(v) \cdot [t < \infty]\,.
\]

\subparagraph*{If $v$ has exactly one child $v'$.}
If $\iota = \false$ and $v \notin S$, then $v$ must be reached through $v'$ at time~$t$.
In this case the minimal total weight reached in $T_{v'}$ is
\[
	D_1 := \min_{t' \leq t} D[v', k-\kappa(t', t), t', \false],
\]
where $\kappa(t', t)$~is the minimal number of delays that need to occur on the edge~$\{v, v'\}$
to ensure that $(v', t')$~reaches~$v$ at time~$t$ but not earlier.
(Set~$\kappa(t', t) = \infty$ if this is impossible.)
Consequently, if $v \notin S$, then
\[
D[v,k,t,\false] = w(v) \cdot  [t < \infty]  + D_1 \,.
\]

If $\iota = \true$ or $v \in S$, then there are two possibilities.
If $v'$ is reached through $v$ at time~$t'$, with $t'$~being the first time $v'$ is reached from~$S$,
then the minimal total weight reached in $T_{v'}$ is 
\[
	D_2 := \min_{t' \geq t} D[v', k-\kappa(t, t'), t', \true] \,.
\]
Otherwise, $v'$ must be reached from a source in $S \cap V(T_{v'})$ at time~$t'$, thus
the minimal total weight reached in~$T_v'$ is
\[D_3 := \min_{t'} D[v', k-\hat{\kappa}(t', t), t', \false],\]
where $\hat{\kappa}(t', t)$ is the minimal number of delays that need to occur on $\{v, v'\}$ to ensure that
$(v, t)$ cannot reach~$v'$ before time~$t'$ and
$(v', t')$ cannot reach~$v$ before time~$t$.
(Again, set $\hat{\kappa}(t', t) = \infty$ if this is impossible.)
Thus we obtain for the case that $\iota = \true$ or $v \in S$ that
\[
	D[v, k, t, \iota] = w(v)\cdot [t < \infty]  + \min \{D_2, D_3\} \,.
\]

\subparagraph*{If $v$ has two children $v', v''$.}
The situation is similar to that of only one child vertex, although more possible cases have to be distinguished.
We omit the tedious details.
However, it is clear that $D[v, k, t, \iota]$ can be computed by simply trying all possible tuples $(t', t'', k', k'', i', i'', \iota', \iota'')$
where $t', t''$ are the times at which $v', v''$ are reached; $k', k''$ are the number of delays occurring in $T_{v'}, T_{v''}$;
$i', i''$ are the number of delays occurring on the edges $\{v, v'\}, \{v,
v''\}$; and $\iota', \iota''$ describe whether $v', v''$ are reached from a
source in their respective subtrees at time $t'$ and~$t''$, respectively.
The number of such tuples and the time required to process each of them is clearly polynomial in $t + k + \abs{\TG}$.

\bigskip{}
After having computed all entries $D[v, k, t, \iota]$, the solution of the \TemporalDelaying{} instance $(\TG, k)$ can be found as the value of 
\[ 
	R[\hat{v}, k] := \min_{t} D[\hat{v}, k, t, \false],
\]
where $\hat{v}$ is the root vertex of $\TG$.

It remains to consider the case that the underlying graph of $\TG$ is a disconnected forest.
In this case simply apply the above algorithm to each connected component.
Afterwards, determining the optimal way to split the overall budget between the connected components can be computed by a simple dynamic program.
Define $X[i, k]$ as the minimum weight reached in the first $i$~trees if up to $k$~time-edges are delayed and use the fact that
\[
X[1, k] = \min_{k' \leq k} \left( R[\hat{v}_1, k']\right)
\]
and for all $i >1$ 
\[
X[i, k] = \min_{k' \leq k} \left( R[\hat{v}_i, k'] + X[i-1, k-k']\right),
\]
where $\hat{v}_j$ is the root of the $j$\textsuperscript{th}~tree.
\end{proof}

\section{Conclusion}
While both problem variants, \TemporalDeleting{} and \TemporalDelaying{},
are polynomial-time solvable on forests and \Wone-hard when parameterized by
$k$, even if the lifetime is $\lifetime=2$, their complexities diverge
when we parameterize by the number~$r$ of reachable vertices.
Here, \TemporalDeleting{} is \Wone-hard while for \TemporalDelaying{} we found
a fixed-parameter tractable algorithm.
This makes \TemporalDelaying{} particularly interesting for applications where the number of reachable vertices should be very small,
e.g.\ when trying to contain the spread of dangerous diseases.

On the practical side we want to point out that
our algorithm for \TemporalDelaying{} parameterized by $r$ uses only linear
space, and its search-tree-based approach makes it fit for optimization techniques like further data reduction rules or pruning using lower bounds.
Furthermore, our max-flow-based branching technique can be turned into a
$r$-approximation for for minimizing the number $r$ of
reachable vertices by delaying $k$ time-edges.
To do so, instead of branching into all choices of $v \in Y$ in line~\ref{line:9} of \cref{algo:search-tree},
simply invoke $g(R \cup Y)$.
Refining the presented technique towards better approximation guarantees seems
to be a promising research direction.
Moreover, when focusing on specific applications, it is natural to exploit
application-dependent graph properties towards designing
more efficient algorithms. In particular: which well-motivated temporal graph
classes beyond trees allow e.g.\ polynomial-time solvability of
\TemporalDeleting{} or \TemporalDelaying{}? Finally, from the viewpoint of
parameterized complexity the parameters~$k$ and $r$ are settled, but the
landscape of structural parameters is still waiting to be explored.

\bibliography{strings-long,bibliography}

\end{document}